\documentclass[12pt]{article}
\usepackage[margin=1in]{geometry}
\usepackage[small]{titlesec}
\usepackage{natbib,amsmath,amssymb,amsthm,graphicx,setspace,paralist,booktabs,rotating,subcaption,times}
\AtBeginDocument{%
\abovedisplayskip=5pt plus 2pt minus 2pt
\belowdisplayskip=\abovedisplayskip
\abovedisplayshortskip=2pt plus 2pt minus 2pt
\belowdisplayshortskip=\belowdisplayskip}
\titlespacing*{\section}{0pt}{*2}{*1}
\titlespacing*{\subsection}{0pt}{*2}{*1}
\bibpunct{(}{)}{;}{a}{}{,}
\newtheorem{proposition}{Proposition}
\newtheorem{theorem}{Theorem}
\newtheorem{lemma}{Lemma}
\theoremstyle{definition}
\newtheorem{condition}{Condition}

\newcommand\relphantom[1]{\mathrel{\phantom{#1}}}
\DeclareMathOperator\sgn{sgn}
\DeclareMathOperator\tr{tr}
\DeclareMathOperator\Cov{Cov}
\DeclareMathOperator\Var{Var}
\DeclareMathOperator\CV{CV}

\DeclareMathOperator\diag{diag}
\DeclareMathOperator*\argmin{\arg\min}
\def\bA{\mathbf{A}}
\def\bB{\mathbf{B}}
\def\bF{\mathbf{F}}
\def\bI{\mathbf{I}}
\def\bQ{\mathbf{Q}}
\def\bS{\mathbf{S}}
\def\bT{\mathbf{T}}
\def\bU{\mathbf{U}}
\def\bW{\mathbf{W}}
\def\bX{\mathbf{X}}
\def\bY{\mathbf{Y}}
\def\bx{\mathbf{x}}
\def\bGamma{\boldsymbol{\Gamma}}
\def\bOmega{\boldsymbol{\Omega}}

\def\balpha{\boldsymbol{\alpha}}
\def\bgamma{\boldsymbol{\gamma}}
\def\bomega{\boldsymbol{\omega}}
\def\bmu{\boldsymbol{\mu}}
\def\bzero{\mathbf{0}}
\def\bone{\mathbf{1}}
\def\cA{\mathcal{A}}

\def\cU{\mathcal{U}}
\def\what{\widehat}

\def\ve{\varepsilon}

\allowdisplaybreaks[1]

\begin{document}

\begin{titlepage}
\setstretch{1.24}
\title{Large Covariance Estimation for Compositional Data via Composition-Adjusted Thresholding}
\author{Yuanpei Cao, Wei Lin, and Hongzhe Li}
\date{}
\maketitle
\thispagestyle{empty}

\footnotetext{Yuanpei Cao is Ph.D. Candidate (E-mail: \emph{yuanpeic@sas.upenn.edu}) and Hongzhe Li is Professor (E-mail: \emph{hongzhe@upenn.edu}), Department of Biostatistics and Epidemiology, Perelman School of Medicine, University of Pennsylvania, Philadelphia, PA 19104. Wei Lin is Assistant Professor, School of Mathematical Sciences and Center for Statistical Science, Peking University, Beijing 100871, China (E-mail: \emph{weilin@math.pku.edu.cn}). Cao and Li's research was supported in part by NIH grants CA127334 and GM097505. Lin's research was supported in part by NSFC grant 71532001, the Recruitment Program of Global Experts, and a start-up grant from Peking University.}

\begin{abstract}
High-dimensional compositional data arise naturally in many applications such as metagenomic data analysis. The observed data lie in a high-dimensional simplex, and conventional statistical methods often fail to produce sensible results due to the unit-sum constraint. In this article, we address the problem of covariance estimation for high-dimensional compositional data, and introduce a composition-adjusted thresholding (COAT) method under the assumption that the basis covariance matrix is sparse. Our method is based on a decomposition relating the compositional covariance to the basis covariance, which is approximately identifiable as the dimensionality tends to infinity. The resulting procedure can be viewed as thresholding the sample centered log-ratio covariance matrix and hence is scalable for large covariance matrices. We rigorously characterize the identifiability of the covariance parameters, derive rates of convergence under the spectral norm, and provide theoretical guarantees on support recovery. Simulation studies demonstrate that the COAT estimator outperforms some naive thresholding estimators that ignore the unique features of compositional data. We apply the proposed method to the analysis of a microbiome dataset in order to understand the dependence structure among bacterial taxa in the human gut.\bigskip

\end{abstract}
\end{titlepage}

\section{Introduction}
Compositional data, which represent the proportions or fractions of a whole, arise naturally in a wide range of applications; examples include geochemical compositions of rocks, household patterns of expenditures, species compositions of biological communities, and topic compositions of documents, among many others. This article is particularly motivated by the metagenomic analysis of microbiome data. The human microbiome is the totality of all microbes at various body sites, whose importance in human health and disease has increasingly been recognized. Recent studies have revealed that microbiome composition varies based on diet, health, and the environment \citep{Huma:fram:2012}, and may play a key role in complex diseases such as obesity, atherosclerosis, and Crohn's disease \citep{Turn:Hama:Yats:Cant:Dunc:core:2009,Koet:Wang:Levi:Buff:Org:inte:2013,Lewi:Chen:Wu:Bush:infl:2015}.

With the development of next-generation sequencing technologies, it is now possible to survey the microbiome composition using direct DNA sequencing of either marker genes or the whole metagenomes. After aligning these sequence reads to the reference microbial genomes, one can quantify the relative abundances of microbial taxa. These sequencing-based microbiome studies, however, only provide a relative, rather than absolute, measure of the abundances of community components. The counts comprising these data (e.g., 16S rRNA gene reads or shotgun metagenomic reads) are set by the amount of genetic material extracted from the community or the sequencing depth, and analysis typically begins by normalizing the observed data by the total number of counts. The resulting fractions thus fall into a class of high-dimensional compositional data that we focus in this article. The high dimensionality refers to the fact that the number of taxa may be comparable to or much larger than the sample size.

An important question in metagenomic studies is to understand the co-occurrence and co-exclusion relationship between microbial taxa, which would provide valuable insights into the complex ecology of microbial communities \citep{Faus:Sath:Izar:Sega:Geve:micr:2012}. Standard correlation analysis from the raw proportions, however, can lead to spurious results due to the unit-sum constraint; the proportions tend to be correlated even if the absolute abundances are independent. Such undesired effects should be removed in an analysis in order to make valid inferences about the underlying biological processes. The compositional effects are further magnified by the low diversity of microbiome data, that is, a few taxa make up the overwhelming majority of the microbiome \citep{Frie:Alm:infe:2012}.

Let $\bX=(X_1,\dots,X_p)^T$ be a composition of $p$ components (taxa) satisfying the simplex constraint
\[
X_j>0,\quad j=1,\dots,p,\quad\sum_{j=1}^pX_j=1.
\]
Owing to the difficulties arising from the simplex constraint, it has been a long-standing question how to appropriately model, estimate, and interpret the covariance structure of compositional data. The pioneering work of \citet{Aitc:stat:1982,Aitc:stat:2003} introduced several equivalent matrix specifications of compositional covariance structures via the log-ratios of components. Statistical methods based on these covariance models respect the unique features of compositional data and prove useful in a variety of applications such as geochemical analysis. A potential disadvantage of these models, however, is that they lack a direct interpretation in the usual sense of covariances and correlations; as a result, it is unclear how to impose certain structures such as sparsity in high dimensions, which is crucial for our applications to microbiome data analysis.

Covariance matrix estimation is of fundamental importance in high-dimensional data analysis and has attracted much recent interest. It is well known that the sample covariance matrix performs poorly in high dimensions and regularization is thus indispensable. \citet{Bick:Levi:cova:2008} and \citet{ElK:oper:2008} introduced regularized estimators by hard thresholding for large covariance matrices that satisfy certain notions of sparsity. \citet{Roth:Levi:Zhu:gene:2009} considered a more general class of thresholding functions, and \citet{Cai:Liu:adap:2011} proposed adaptive thresholding that adapts to the variability of individual entries. Exploiting a factor model structure, \citet{Fan:Fan:Lv:high:2008} proposed a factor-based method for high-dimensional covariance matrix estimation. \citet{Fan:Liao:Minc:larg:2013} extended the work by considering a conditional sparsity structure and developed a POET method by thresholding principal orthogonal complements.

In this article, we address the problem of covariance estimation for high-dimensional compositional data. Let $\bW=(W_1,\dots,W_p)^T$ with $W_j>0$ for all $j$ be a vector of latent variables, called the \emph{basis}, that generate the observed data via the normalization
\begin{equation}\label{eq:norm}
X_j=\frac{W_j}{\sum_{i=1}^pW_i},\quad j=1,\dots,p.
\end{equation}
Estimating the covariance structure of $\bW$ has traditionally been considered infeasible owing to the apparent lack of identifiability. By exploring a decomposition relating the compositional covariance to the basis covariance, we find, however, that the nonidentifiability vanishes asymptotically as the dimensionality grows under certain sparsity assumptions. More specifically, define the \emph{basis covariance matrix} $\bOmega_0=(\omega_{ij}^0)_{p\times p}$ by
\begin{equation}\label{eq:omega}
\omega_{ij}^0=\Cov(Y_i,Y_j),
\end{equation}
where $Y_j=\log W_j$. Then $\bOmega_0$ is approximately identifiable as long as it belongs to a class of large sparse covariance matrices.

The somewhat surprising ``blessing of dimensionality'' allows us to develop a simple, two-step method by first extracting a rank-2 component from the decomposition and then estimating the sparse component $\bOmega_0$ by thresholding the residual matrix. The resulting procedure can equivalently be viewed as thresholding the sample centered log-ratio covariance matrix, and hence is optimization-free and scalable for large covariance matrices. We call our method \emph{composition-adjusted thresholding} (COAT), which removes the ``coat'' of compositional effects from the covariance structure. We derive rates of convergence under the spectral norm and provide theoretical guarantees on support recovery. Simulation studies demonstrate that the COAT estimator outperforms some naive thresholding estimators that ignore the unique features of compositional data. We illustrate our method by analyzing a microbiome dataset in order to understand the dependence structure among bacterial taxa in the human gut.

The covariance relationship, which was due to \citet[sec.~4.11]{Aitc:stat:2003}, has recently been exploited to develop algorithms for inferring correlation networks from metagenomic data \citep{Frie:Alm:infe:2012,Fang:Huan:Zhao:Deng:ccla:2015,Ban:An:Jiang:inve:2015}. Our contributions here are to turn the idea into a principled approach to sparse covariance matrix estimation and provide statistical insights into the issue of identifiability and the impacts of dimensionality. Our method also bears some resemblance to the POET method proposed by \citet{Fan:Liao:Minc:larg:2013} in that underlying both methods is a low-rank plus sparse matrix decomposition. The rank-2 component in our method, however, arises from the covariance structure of compositional data rather than a factor model assumption. As a result, it can be obtained by simple algebraic operations without computing the principal components.

The rest of the article is organized as follows. Section 2 reviews a covariance relationship and addresses the issue of identifiability. Section 3 introduces the COAT methodology. Section 4 investigates the theoretical properties of the COAT estimator in terms of convergence rates and support recovery. Simulation studies and an application to human gut microbiome data are presented in Sections 5 and 6, respectively. We conclude the article with some discussion in Section 7 and relegate all proofs to the Appendix.

\section{Identifiability of the Covariance Model}
We first introduce some notation. Denote by $\|\cdot\|_1$, $\|\cdot\|_2$, $\|\cdot\|_F$, and $\|\cdot\|_{\max}$ the matrix $L_1$-norm, spectral norm, Frobenius norm, and entrywise $L_\infty$-norm, defined for a matrix $\bA=(a_{ij})$ by $\|\bA\|_1=\max_j\sum_i|a_{ij}|$, $\|\bA\|_2=\sqrt{\lambda_{\max}(\bA^T\bA)}$, $\|\bA\|_F=\sqrt{\sum_{i,j}a_{ij}^2}$, and $\|\bA\|_{\max}=\max_{i,j}|a_{ij}|$, where $\lambda_{\max}(\cdot)$ denotes the largest eigenvalue.

In the latent variable covariance model \eqref{eq:norm} and \eqref{eq:omega}, the basis covariance matrix $\bOmega_0$ is the parameter of interest. One of the matrix specifications of compositional covariance structures introduced by \citet{Aitc:stat:2003} is the \emph{variation matrix} $\bT_0=(\tau_{ij}^0)_{p\times p}$ defined by
\begin{equation}\label{eq:tau}
\tau_{ij}^0=\Var(\log(X_i/X_j)).
\end{equation}
In view of the relationship \eqref{eq:norm}, we can decompose $\tau_{ij}^0$ as
\begin{align}
\tau_{ij}^0&=\Var(\log W_i-\log W_j)\notag\\
&=\Var(Y_i)+\Var(Y_j)-2\Cov(Y_i,Y_j)\notag\\
&=\omega_{ii}^0+\omega_{jj}^0-2\omega_{ij}^0,\label{eq:decomp_ij}
\end{align}
or in matrix form,
\begin{equation}\label{eq:decomp}
\bT_0=\bomega_0\bone^T+\bone\bomega_0^T-2\bOmega_0,
\end{equation}
where $\bomega_0=(\omega_{11}^0,\dots,\omega_{pp}^0)^T$ and $\bone=(1,\dots,1)^T$. Corresponding to the many-to-one relationship between bases and compositions, the basis covariance matrix $\bOmega_0$ is unidentifiable from the decomposition \eqref{eq:decomp}, since $\bomega_0\bone^T+\bone\bomega_0^T$ and $\bOmega_0$ are in general not orthogonal to each other (with respect to the usual Euclidean inner product). In fact, using the \emph{centered log-ratio covariance matrix} $\bGamma_0=(\gamma_{ij}^0)_{p\times p}$ defined by
\[
\gamma_{ij}^0=\Cov\{\log(X_i/g(\bX)),\log(X_j/g(\bX))\},
\]
where $g(\bx)=(\prod_{j=1}^px_j)^{1/p}$ is the geometric mean of a vector $\bx=(x_1,\dots,x_p)^T$, we can similarly write
\begin{align*}
\tau_{ij}^0&=\Var\{\log(X_i/g(\bX))-\log(X_j/g(\bX))\}\\
&=\Var\{\log(X_i/g(\bX))\}+\Var\{\log(X_j/g(\bX))\}-2\Cov\{\log(X_i/g(\bX),\log(X_j/g(\bX))\}\\
&=\gamma_{ii}^0+\gamma_{jj}^0-2\gamma_{ij}^0,
\end{align*}
or in matrix form,
\begin{equation}\label{eq:decomp_orth}
\bT_0=\bgamma_0\bone^T+\bone\bgamma_0^T-2\bGamma_0,
\end{equation}
where $\bgamma_0=(\gamma_{11}^0,\dots,\gamma_{pp}^0)^T$ and $\bone=(1,\dots,1)^T$. Unlike \eqref{eq:decomp}, the following proposition shows that \eqref{eq:decomp_orth} is an orthogonal decomposition and hence the components $\bgamma_0\bone^T+\bone\bgamma_0^T$ and $\bGamma_0$ are identifiable. In addition, by comparing the decompositions \eqref{eq:decomp} and \eqref{eq:decomp_orth}, we can bound the difference between $\bOmega_0$ and its identifiable counterpart $\bGamma_0$ as follows.

\begin{proposition}\label{prop:ident}
The components $\bgamma_0\bone^T+\bone\bgamma_0^T$ and $\bGamma_0$ in the decomposition \eqref{eq:decomp_orth} are orthogonal to each other. Moreover, for the covariance parameters $\bOmega_0$ and $\bGamma_0$ in the decompositions \eqref{eq:decomp} and \eqref{eq:decomp_orth},
\[
\|\bOmega_0-\bGamma_0\|_{\max}\le3p^{-1}\|\bOmega_0\|_1.
\]
\end{proposition}

Proposition \ref{prop:ident} entails that the covariance parameter $\bOmega_0$ is \emph{approximately} identifiable as long as $\|\bOmega_0\|_1=o(p)$. In particular, suppose that $\bOmega_0$ belongs to a class of sparse covariance matrices considered by \citet{Bick:Levi:cova:2008},
\begin{equation}\label{eq:class}
\cU(q,s_0(p),M)\equiv\left\{\bOmega\colon\bOmega\succ0,\max_j\omega_{jj}\le M,\max_i\sum_{j=1}^p|\omega_{ij}|^q\le s_0(p)\right\},
\end{equation}
where $0\le q<1$ and $\bOmega\succ0$ denotes that $\bOmega$ is positive definite. Then
\[
\|\bOmega_0\|_1=\max_i\sum_{j=1}^p|\omega_{ij}^0|^{1-q}|\omega_{ij}^0|^q\le\max_i\sum_{j=1}^p(\omega_{ii}^0\omega_{jj}^0)^{(1-q)/2}|\omega_{ij}^0|^q\le M^{1-q}s_0(p),
\]
and hence the parameters $\bOmega_0$ and $\bGamma_0$ are asymptotically indistinguishable when $s_0(p)=o(p)$. This allows us to use $\bGamma_0$ as a proxy for $\bOmega_0$ and greatly facilitates the development of new methodology and associated theory. The intuition behind the approximate identifiability under the sparsity assumption is that the rank-2 component $\bomega_0\bone^T+\bone\bomega_0^T$ represents a global effect that spreads across all rows and columns, while the sparse component $\bOmega_0$ represents a local effect that is confined to individual entries.

Also of interest is the \emph{exact} identifiability of $\bOmega_0$ over $L_0$-balls, which has been studied by \citet{Fang:Huan:Zhao:Deng:ccla:2015} and \citet{Ban:An:Jiang:inve:2015}. The following result provides a sufficient and necessary condition for the exact identifiability of $\bOmega_0$ by confining it to an $L_0$-ball.

\begin{proposition}\label{prop:ident_exact}
Suppose that $\bOmega_0$ belongs to the $L_0$-ball
\[
\mathcal{B}_0(s_e(p))\equiv\left\{\bOmega\colon\sum_{(i,j)\colon i<j}I(\omega_{ij}\ne0)\le s_e(p)\right\},
\]
where $p\ge5$. Then there exist no two values of $\bOmega_0$ that correspond to the same $\bT_0$ in \eqref{eq:decomp} if and only if $s_e(p)<(p-1)/2$.
\end{proposition}

A counterexample is provided in the proof of Proposition \ref{prop:ident_exact} to show that the sparsity conditions in \citet{Fang:Huan:Zhao:Deng:ccla:2015} and \citet{Ban:An:Jiang:inve:2015}, which are both at the order of $O(p^2)$, do not suffice. The identifiability condition in Proposition \ref{prop:ident_exact} essentially requires the average degree of the correlation network to be less than 1, which is too restrictive to be useful in practice. This illustrates the importance and necessity of introducing the notion of approximate identifiability.

\section{A Sparse Covariance Estimator for Compositional Data}
Suppose that $(\bW_k,\bX_k)$, $k=1,\dots,n$, are independent copies of $(\bW,\bX)$, where the compositions $\bX_k=(X_{k1},\dots,X_{kp})^T$ are observed and the bases $\bW_k=(W_{k1},\dots,W_{kp})^T$ are latent. In Section 3.1, we rely on the decompositions \eqref{eq:decomp} and \eqref{eq:decomp_orth} and Proposition \ref{prop:ident} to develop an estimator of $\bOmega_0$, and in Section 3.2 discuss the selection of the tuning parameter.

\subsection{Composition-Adjusted Thresholding}
In view of Proposition \ref{prop:ident}, we wish to estimate the covariance parameter $\bOmega_0$ via the proxy $\bGamma_0$. To this end, we first construct an empirical estimate of $\bGamma_0$ and then apply adaptive thresholding to the estimate.

There are two equivalent ways to form the estimate of $\bGamma_0$. Motivated by the decomposition \eqref{eq:decomp_orth}, one can start with the sample counterpart $\what\bT=(\hat\tau_{ij})_{p\times p}$ of $\bT_0$ defined by
\[
\hat\tau_{ij}=\frac{1}{n}\sum_{k=1}^n(\tau_{kij}-\bar\tau_{ij})^2,
\]
where $\tau_{kij}=\log(X_{ki}/X_{kj})$ and $\bar\tau_{ij}=n^{-1}\sum_{k=1}^n\tau_{kij}$. A rank-2 component $\what\balpha\bone^T+\bone\what\balpha^T$ with $\what\balpha=(\hat\alpha_1,\dots,\hat\alpha_p)^T$ can be extracted from the decomposition \eqref{eq:decomp_orth} by projecting $\what\bT$ onto the subspace $\cA\equiv\{\balpha\bone^T+\bone\balpha^T\colon\balpha\in\mathbb{R}^p\}$, which is given by
\[
\hat\alpha_i=\hat\tau_{i\cdot}-\frac{1}{2}\hat\tau_{\cdot\cdot},
\]
where $\hat\tau_{i\cdot}=p^{-1}\sum_{j=1}^p\hat\tau_{ij}$ and $\hat\tau_{\cdot\cdot}=p^{-2}\sum_{i,j=1}^p\hat\tau_{ij}$. The residual matrix $\what\bGamma=-(\what\bT-\what\balpha\bone^T-\bone\what\balpha^T)/2$, with entries
\[
\hat\gamma_{ij}=-\frac{1}{2}(\hat\tau_{ij}-\hat\alpha_i-\hat\alpha_j)=-\frac{1}{2}(\hat\tau_{ij}-\hat\tau_{i\cdot}-\hat\tau_{j\cdot} +\hat\tau_{\cdot\cdot}),
\]
is then an estimate of $\bGamma_0$. Alternatively, $\what\bGamma$ can be obtained directly as the sample counterpart of $\bGamma_0$ through the expression
\begin{equation}\label{eq:gamma}
\hat\gamma_{ij}=\frac{1}{n}\sum_{k=1}^n(\gamma_{ki}-\bar\gamma_i)(\gamma_{kj}-\bar\gamma_j),
\end{equation}
where $\gamma_{kj}=\log(X_{kj}/g(\bX_k))$ and $\bar\gamma_j=n^{-1}\sum_{k=1}^n\gamma_{kj}$.

Now applying adaptive thresholding to $\what\bGamma$, we define the \emph{composition-adjusted thresholding} (COAT) estimator
\begin{equation}\label{eq:coat}
\what\bOmega=(\hat\omega_{ij})_{p\times p}\quad\text{with }\hat\omega_{ij}=S_{\lambda_{ij}}(\hat\gamma_{ij}),
\end{equation}
where $S_\lambda(\cdot)$ is a general thresholding function and $\lambda_{ij}>0$ are entry-dependent thresholds.

In this article, we consider a class of general thresholding functions $S_\lambda(\cdot)$ that satisfy the following conditions:
\begin{compactenum}[(i)]
  \item $S_\lambda(z)=0$ for $|z|\le\lambda$;
  \item $|S_\lambda(z)-z|\le\lambda$ for all $z\in\mathbb{R}$.
\end{compactenum}
These two conditions were assumed by \citet{Roth:Levi:Zhu:gene:2009} and \citet{Cai:Liu:adap:2011} along with another condition that is not required in our analysis. Examples of thresholding functions belonging to this class include the hard thresholding rule $S_\lambda(z)=zI(|z|\ge\lambda)$, the soft thresholding rule $S_\lambda(z)=\sgn(z)(|z|-\lambda)_+$, and the adaptive lasso rule $S_\lambda(z)=z(1-|\lambda/z|^\eta)_+$ for $\eta\ge1$.

The performance of the COAT estimator depends critically on the choice of thresholds. Using entry-adaptive thresholds may in general improve the performance over applying a universal threshold. To derive a data-driven choice of $\lambda_{ij}$, define
\[
\theta_{ij}=\Var\{(Y_i-\mu_i)(Y_j-\mu_j)\},
\]
where $\mu_j=EY_j$. We take $\lambda_{ij}$ to be of the form
\begin{equation}\label{eq:thresh}
\lambda_{ij}=\lambda\sqrt{\hat\theta_{ij}},
\end{equation}
where $\hat\theta_{ij}$ are estimates of $\theta_{ij}$, and $\lambda>0$ is a tuning parameter to be chosen, for example, by cross-validation. We rewrite \eqref{eq:gamma} as $\hat\gamma_{ij}=n^{-1}\sum_{k=1}^n\gamma_{kij}$, where $\gamma_{kij}=(\gamma_{ki}-\bar\gamma_i)(\gamma_{kj}-\bar\gamma_j)$. Then $\theta_{ij}$ can be estimated by
\[
\hat\theta_{ij}=\frac{1}{n}\sum_{k=1}^n(\gamma_{kij}-\hat\gamma_{ij})^2.
\]

\subsection{Tuning Parameter Selection}
The thresholds defined by \eqref{eq:thresh} depend on the tuning parameter $\lambda$, which can be chosen through $V$-fold cross-validation. Denote by $\what\bOmega^{(-v)}(\lambda)$ the COAT estimate based on the training data excluding the $v$th fold, and $\what\bGamma_v$ the residual matrix (or the sample centered log-ratio covariance matrix) based on the test data including only the $v$th fold. We choose the optimal value of $\lambda$ that minimizes the cross-validation error
\[
\CV(\lambda)=\frac{1}{V}\sum_{v=1}^V\|\what\bOmega^{(-v)}(\lambda)-\what\bGamma^{(v)}\|_F^2.
\]
With the optimal $\lambda$, we then compute the COAT estimate based on the full dataset as our final estimate. When the positive definiteness of the covariance estimate in finite samples is required for interpretation, we follow the approach of \citet{Fan:Liao:Minc:larg:2013} and choose $\lambda$ in the range where the minimum eigenvalue of the COAT estimate is positive.

\section{Theoretical Properties}
In this section, we investigate the asymptotic properties of the COAT estimator. As a distinguishing feature of our theoretical analysis, we assume neither the exact identifiability of the parameters nor that the degree of (approximate) identifiability is dominated by the statistical error. Instead, the degree of identifiability enters our analysis and shows up in the resulting rate of convergence. Such theoretical analysis is rare in the literature, but is extremely relevant for latent variable models in the presence of nonidentifiability and is of theoretical interest in its own right. We introduce our assumptions in Section 4.1, and present our main results on rates of convergence and support recovery in Section 4.2.

\subsection{Assumptions}
Recall that $Y_j=\log W_j$, $\mu_j=EY_j$, and $\theta_{ij}=\Var\{(Y_i-\mu_i)(Y_j-\mu_j)\}$, and define $Y_{kj}=\log W_{kj}$. Without loss of generality, assume $\mu_j=0$ for all $j$ throughout this section. We need to impose the following moment conditions on the log-basis $\bY=(Y_1,\dots,Y_p)^T$.

\begin{condition}\label{cond:tail}
There exists a constant $\alpha>0$ such that $\max_jE\exp(\alpha Y_j^2)\le2$.
\end{condition}

\begin{condition}\label{cond:sparse}
The basis covariance matrix $\bOmega_0$ belongs to the class $\cU(q,s_0(p),M)$ defined by \eqref{eq:class}, where $0\le q<1$, $s_0(p)=o(p)$, and $\log p=o(n^{1/5})$.
\end{condition}

\begin{condition}\label{cond:lower}
There exists a constant $\tau>0$ such that $\min_{i,j}\theta_{ij}\ge\tau$.
\end{condition}

\begin{condition}\label{cond:moment4}
There exists a sequence $s_1(p)=o(p)$ such that
\[
\max_{i,j,\ell}\left|\sum_{m=1}^pEY_iY_jY_{\ell}Y_m\right|\le s_1(p).
\]
\end{condition}

Conditions 1--3 are similar to those commonly assumed in the covariance estimation literature; see, for example, \citet{Cai:Liu:adap:2011}. Condition \ref{cond:tail} requires that the variables $Y_j$s be uniformly sub-Gaussian; the definition we use here is among several equivalent ways of defining sub-Gaussianity \citep[][sec.\ 2.3]{Bouc:Lugo:Mass:conc:2013}, and is most convenient for our technical analysis. Condition \ref{cond:sparse} imposes some restrictions on the dimensionality and sparsity of the basis covariance matrix $\bOmega_0$. It is worth mentioning that the sparsity level condition $s_0=o(p)$ is so weak that it suffices to guarantee only approximate identifiability but allows the degree of nonidentifiability to be large relative to the statistical error. Condition \ref{cond:lower} is essential for methods based on adaptive thresholding. Condition \ref{cond:moment4} arises from identifiability considerations in estimating the variances $\theta_{ij}$. In particular, if $\bY$ is multivariate normal, then Condition \ref{cond:moment4} is implied by the assumptions $\bOmega_0\in\cU(q,s_0(p),M)$ and $s_0(p)=o(p)$ in Condition \ref{cond:sparse}, since from Isserlis' theorem \citep{Isse:on:1918} we have
\[
\max_{i,j,\ell}\left|\sum_{m=1}^pEY_iY_jY_{\ell}Y_m\right|\le\max_{i,j,\ell}\sum_{m=1}^p\left(|\omega_{ij}^0||\omega_{\ell m}^0| +|\omega_{i\ell}^0||\omega_{jm}^0|+|\omega_{im}^0||\omega_{j\ell}^0|\right)\le 3M^{2-q}s_0(p).
\]

\subsection{Main Results}
We are now in a position to state our main results. The following theorem gives the rate of convergence under the spectral norm for the COAT estimator.
\begin{theorem}[Rate of convergence]\label{thm:rate}
Under Conditions \ref{cond:tail}--\ref{cond:moment4}, if the tuning parameter $\lambda$ in \eqref{eq:thresh} is chosen to be
\begin{equation}\label{eq:lambda}
\lambda=C_1\sqrt{\frac{\log p}{n}}+C_2\frac{s_0(p)}{p}
\end{equation}
for sufficiently large $C_1,C_2>0$, then the COAT estimator $\what\bOmega$ in \eqref{eq:coat} satisfies
\[
\|\what\bOmega-\bOmega_0\|_2=O_p\left\{s_0(p)\left(\sqrt{\frac{\log p}{n}}+\frac{s_0(p)}{p}\right)^{1-q}\right\}
\]
uniformly on $\cU(q,s_0(p),M)$.
\end{theorem}

The rate of convergence provided by Theorem \ref{thm:rate} exhibits an interesting decomposition: the term $s_0(p)\{(\log p)/n\}^{(1-q)/2}$ represents the estimation error due to estimating $\bGamma_0$, while the term $s_0(p)(s_0(p)/p)^{1-q}$ accounts for the approximation error due to using $\bGamma_0$ as a proxy for $\bOmega_0$. In particular, if the approximation error is dominated by the estimation error, then the COAT estimator attains the minimax optimal rate under the spectral norm over $\cU(q,s_0(p),M)$ \citep{Cai:Zhou:opti:2012}. It is important to note that the dimensionality $p$ appears in both terms where it plays opposite roles. We observe a ``curse of dimensionality'' in the first term, where the growth of dimensionality contributes a logarithmic factor to the estimation error. In contrast, a ``blessing of dimensionality'' is reflected by the second term in that a diverging dimensionality shrinks the approximation error toward zero at a power rate.

The insights gained from Theorem \ref{thm:rate} have important implications for compositional data analysis. In the analysis of many compositional datasets, the dimensionality often depends on the taxonomic level to be examined. For example, in metagenomic studies, the dimensionality may range from only a few taxa at the phylum level to thousands of taxa at the operational taxonomic unit (OTU) level. Suppose, for simplicity, that the magnitudes of correlation signals are of about the same order across different taxonomic levels. Then Theorem \ref{thm:rate} indicates a tradeoff between an accurate estimation of the covariance structure with low dimensionality and a sensible interpretation in terms of the basis components with high dimensionality. This tradeoff thus suggests the need to analyze compositional data at relatively finer taxonomic levels when a latent variable interpretation is desired.

The proof of Theorem \ref{thm:rate} relies on a series of concentration inequalities that take the approximation error term into account, which can be found in the Appendix. As a consequence of these inequalities, we obtain the following result regarding the support recovery property of the COAT estimator. Here the support of $\bOmega_0$ refers to the set of all indices $(i,j)$ with $\omega_{ij}^0\ne 0$.

\begin{theorem}[Support recovery]\label{thm:supp}
Under Conditions \ref{cond:tail}--\ref{cond:moment4}, if the tuning parameter $\lambda$ in \eqref{eq:thresh} is chosen as in \eqref{eq:lambda}, then the COAT estimator $\what\bOmega$ in \eqref{eq:coat} satisfies
\begin{equation}\label{eq:sparsist}
P\left(\hat\omega_{ij}=0\text{ for all }(i,j)\text{ with }\omega_{ij}^0=0\right)\to 1.
\end{equation}
Moreover, if in addition
\begin{equation}\label{eq:min_sig}
\min_{(i,j)\colon\omega_{ij}^0\ne 0}|\omega_{ij}^0|/\sqrt{\theta_{ij}}\ge C\lambda
\end{equation}
for some constant $C>3/2$, then
\begin{equation}\label{eq:supp}
P\left(\sgn(\hat\omega_{ij})=\sgn(\omega_{ij}^0)\text{ for all }(i,j)\right)\to 1.
\end{equation}
\end{theorem}

Theorem \ref{thm:supp} parallels the support recovery results in \citet{Roth:Levi:Zhu:gene:2009} and \citet{Cai:Liu:adap:2011}. However, owing to the extra term $s_0(p)/p$ in the expression of $\lambda$, the assumption \eqref{eq:min_sig} requires in addition that no correlation signals fall below the approximation error. In other words, exact support recovery will break down if any correlation signal is confounded by the compositional effect.

\section{Simulation Studies}
We conducted simulation studies to compare the numerical performance of the COAT estimator $\what\bOmega$ with that of the oracle thresholding estimator $\what\bOmega_o$, which knew the latent basis components and applied the thresholding procedure to the sample covariance matrix of the log-basis $\bY$. We also include in our comparison two naive thresholding estimators $\what\bOmega_c$ and $\what\bOmega_l$, which are based on the sample covariance matrices of the composition $\bX$ and its logarithm $\log\bX$, respectively. Note that $\what\bOmega_o$ is the ideal estimator that the COAT estimator attempts to mimic, whereas both $\what\bOmega_c$ and $\what\bOmega_l$ ignore the unique features of compositional data and thus are expected to perform poorly.

\subsection{Simulation Settings}
The data $(\bW_k,\bX_k)$, $k=1,\dots,n$, were generated as follows. We first generated $\bY_k$ in two different ways:
\begin{compactenum}[(i)]
  \item $\bY_k$ are independent from the multivariate normal distribution $N_p(\bmu,\bOmega_0)$;
  \item $\bY_k=\bmu+\bF\bU_k/\sqrt{10}$, where $\bF\bF^T=\bOmega_0$ and the components of $\bU_k$ are independent gamma variables with shape parameter 10 and scale parameter 1, so that $\Var(\bY_k)=\bOmega_0$. Here the matrix $\bF$ is obtained by computing the singular value decomposition $\bOmega_0=\bQ\bS\bQ^T$ and letting $\bF=\bQ\bS^{1/2}$.
\end{compactenum}
Then $\bW_k=(W_{k1},\dots,W_{kp})^T$ and $\bX_k=(X_{k1},\dots,X_{kp})^T$ were obtained through the transformations $W_{kj}=e^{Y_{kj}}$ and $X_{kj}=W_{kj}/\sum_{i=1}^pW_{ki}$, $j=1,\dots,p$. Hence, in Case (i), $\bW_k$ and $\bX_k$ follow multivariate log-normal and logistic normal distributions \citep{Aitc:Shen:logi:1980}, respectively; the distributions of $\bW_k$ and $\bX_k$ in Case (ii) can similarly be viewed as a type of multivariate log-gamma and logistic-gamma distributions.

In both cases, we took the components of $\bmu$ randomly from the uniform distribution on $[0,10]$, in order to reflect the fact that compositional data arising from metagenomic studies are often heterogeneous. The following two models for the covariance matrix $\bOmega_0$ were considered:

\begin{compactitem}
  \item Model 1 (Identity covariance): $\bOmega_0=\bI_p$.
  \item Model 2 (Sparse covariance): $\bOmega_0=\diag(\bA_1,\bA_2)$, where $\bA_1=\bB+\ve\bI_{p_1}$, $\bA_2=4\bI_{p_2}$, $p_1=\lfloor 2\sqrt{p}\rfloor$, $p_2=p-p_1$, and $\bB$ is a symmetric matrix whose lower triangular entries are independent from the uniform distribution on $[-1,-0.5]\cup[0.5,1]$ with probability 0.2 and equal to 0 with probability 0.8. We set $\ve=\max(-\lambda_{\min}(\bB),0)+0.01$ to ensure that $\bA_1$ is positive definite, where $\lambda_{\min}(\cdot)$ denotes the smallest eigenvalue.
\end{compactitem}

Model 1 is an extreme but illustrative case intended for comparing the distributions of spurious correlations under different transformations. The setting of Model 2 is typical in the covariance estimation literature and similar to that in \citet{Cai:Liu:adap:2011}. We set the sample size $n=100$ and the dimension $p=50$, 100, and 200, and repeated 100 simulations for each setting.

\subsection{Spurious Correlations}
The boxplots of sample correlations with simulated data under different transformations in Model 1 are shown in Figure \ref{fig:boxplot}. Clearly, the sample centered log-ratio (clr) correlations are centered around zero and have a similar distribution to that of the sample correlations of $\bY$; the resemblance tends to increase as the dimension $p$ grows. This trend is consistent with Proposition \ref{prop:ident} and provides numerical evidence for the validity of the centered log-ratio covariance matrix $\bGamma_0$ as a proxy for $\bOmega_0$. In fact, from the proof of Proposition \ref{prop:ident} we have, when $\bOmega_0=\bI_p$,
\[
\|\bOmega_0-\bGamma_0\|_{\max}=\max_{i,j}|\omega_{i\cdot}^0+\omega_{j\cdot}^0-\omega_{\cdot\cdot}^0|=p^{-1}.
\]
In contrast, the phenomenon of spurious correlations is observed on both $\log\bX$ and $\bX$. The sample correlations of $\log\bX$ exhibit a severe upward bias, while the sample correlations of $\bX$ contain many outliers that would be detected as signals by a thresholding procedure with threshold level close to 1. Moreover, the spurious correlations seem to become worse with gamma-related distributions where the components of the composition have more heterogeneous means.

\begin{figure}
\includegraphics[width=\textwidth]{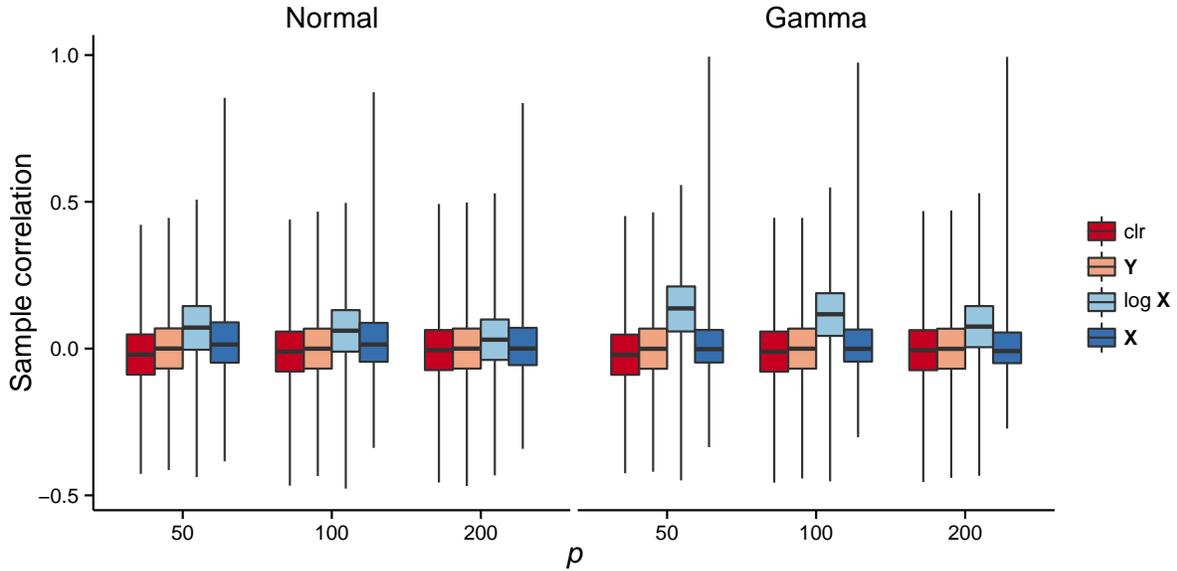}
\caption{Boxplots of sample correlations with simulated data under different transformations in Model 1.}\label{fig:boxplot}
\end{figure}

\subsection{Performance Comparisons}
We applied the COAT method with hard and soft thresholding rules to simulated data in Model 2. For comparison, we also applied the thresholding procedure to the sample covariance matrices of $\bY$, $\log\bX$, and $\bX$, resulting in the estimators $\what\bOmega_o$, $\what\bOmega_l$, and $\what\bOmega_c$, respectively. The tuning parameter $\lambda$ in each thresholding estimator was chosen by tenfold cross-validation. Losses under the matrix $L_1$-norm, spectral norm, and Frobenius norm were used to measure the estimation performance, while the true positive rate and false positive rate were employed to assess the quality of support recovery.

The simulation results for Model 2 with normal- and gamma-related distributions are summarized in Tables \ref{table:normal} and \ref{table:gamma}, respectively. We see that the COAT estimator $\what\bOmega$ performs almost equally well as the ideal estimator $\what\bOmega_o$, and outperforms the naive thresholding estimators $\what\bOmega_l$ and $\what\bOmega_c$ by a large margin. In particular, the estimation losses of $\what\bOmega_l$ are disastrously large in the gamma setting, in agreement with the severe bias observed in Figure \ref{fig:boxplot}. The estimation losses of $\what\bOmega_c$ do not change much across different thresholding rules and distributions, since all entries of the estimate are very small relative to the true values. Both $\what\bOmega_l$ and $\what\bOmega_c$ show inferior performance in terms of true and false positive rates, indicating that they are not model selection consistent. Comparisons between hard and soft thresholding rules suggest that the former is more conservative in selecting false positives and results in a more parsimonious model, whereas the latter strikes a balance between true and false positives due to the shrinkage effect.

\begin{sidewaystable}
\def~{\phantom{0}}
\caption{Means (standard errors) of various performance measures for four methods with hard and soft thresholding rules in Model 2 with normal-related distributions over 100 replications}\label{table:normal}
\begin{tabular*}{\textwidth}{@{}l*{8}{@{\extracolsep{\fill}}c}@{}}
\toprule\toprule
& \multicolumn{4}{c}{Hard} & \multicolumn{4}{c@{}}{Soft}\\
\cmidrule{2-5}\cmidrule{6-9}
$p$ & $\what\bOmega$ & $\what\bOmega_o$ & $\what\bOmega_l$ & $\what\bOmega_c$ & $\what\bOmega$ & $\what\bOmega_o$ & $\what\bOmega_l$ & $\what\bOmega_c$\\ \midrule\addlinespace
\multicolumn{9}{@{}c@{}}{Matrix $L_1$-norm loss}\\
~50 & ~4.09 (0.05) & 4.02 (0.05) & 11.72 (1.51) & ~6.91 (0.00) & ~4.34 (0.05) & ~4.10 (0.05) & 18.73 (0.64) & ~6.91 (0.00)\\
100 & ~5.46 (0.04) & 5.50 (0.05) & ~7.85 (1.13) & ~8.07 (0.00) & ~5.50 (0.05) & ~5.40 (0.05) & 27.10 (1.18) & ~8.07 (0.00)\\
200 & ~8.07 (0.04) & 8.10 (0.04) & ~8.36 (0.04) & 10.93 (0.00) & ~7.72 (0.06) & ~7.66 (0.06) & 22.61 (1.13) & 10.93 (0.00)\\\addlinespace
\multicolumn{9}{@{}c@{}}{Spectral norm loss}\\
~50 & ~2.32 (0.02) & 2.22 (0.02) & ~7.23 (0.99) & ~4.91 (0.00) & ~2.49 (0.02) & ~2.40 (0.02) & 10.23 (0.42) & ~4.92 (0.00)\\
100 & ~2.89 (0.02) & 2.90 (0.02) & ~4.50 (0.74) & ~5.46 (0.00) & ~3.01 (0.02) & ~2.98 (0.02) & 13.93 (0.70) & ~5.46 (0.00)\\
200 & ~3.55 (0.02) & 3.55 (0.02) & ~3.68 (0.02) & ~6.43 (0.00) & ~3.93 (0.02) & ~3.89 (0.02) & ~9.28 (0.60) & ~6.43 (0.00)\\\addlinespace
\multicolumn{9}{@{}c@{}}{Frobenius norm loss}\\
~50 & ~5.63 (0.03) & 5.50 (0.03) & 11.47 (1.01) & 26.00 (0.00) & ~8.37 (0.03) & ~7.99 (0.03) & 15.18 (0.39) & 26.01 (0.00)\\
100 & ~8.70 (0.04) & 8.66 (0.03) & 11.39 (0.81) & 38.39 (0.00) & 13.11 (0.04) & 12.87 (0.04) & 24.18 (0.70) & 38.39 (0.00)\\
200 & 12.03 (0.03) & 12.05(0.03) & 12.97 (0.05) & 55.78 (0.00) & 20.48 (0.03) & 20.32 (0.03) & 27.06 (0.68) & 55.78 (0.00)\\\addlinespace
\multicolumn{9}{@{}c@{}}{True positive rate}\\
~50 & ~0.65 (0.01) & 0.67 (0.01) & ~0.70 (0.01) & ~0.76 (0.02) & ~0.94 (0.00) & ~0.95 (0.00) & ~0.93 (0.00) & ~0.94 (0.00)\\
100 & ~0.59 (0.00) & 0.59 (0.00) & ~0.59 (0.01) & ~0.46 (0.02) & ~0.91 (0.00) & ~0.91 (0.00) & ~0.87 (0.00) & ~0.92 (0.00)\\
200 & ~0.60 (0.00) & 0.60 (0.00) & ~0.60 (0.00) & ~0.36 (0.02) & ~0.83 (0.00) & ~0.84 (0.00) & ~0.87 (0.00) & ~0.89 (0.00)\\\addlinespace
\multicolumn{9}{@{}c@{}}{False positive rate}\\
~50 & ~0.00 (0.00) & 0.00 (0.00) & ~0.15 (0.03) & ~0.44 (0.03) & ~0.11 (0.00) & ~0.09 (0.00) & ~0.53 (0.01) & ~0.61 (0.01)\\
100 & ~0.00 (0.00) & 0.00 (0.00) & ~0.02 (0.01) & ~0.11 (0.01) & ~0.06 (0.00) & ~0.06 (0.00) & ~0.41 (0.01) & ~0.59 (0.01)\\
200 & ~0.00 (0.00) & 0.00 (0.00) & ~0.00 (0.00) & ~0.07 (0.01) & ~0.03 (0.00) & ~0.03 (0.00) & ~0.18 (0.01) & ~0.54 (0.01)\\
\bottomrule
\end{tabular*}
\end{sidewaystable}

\begin{sidewaystable}
\def~{\phantom{0}}
\caption{Means (standard errors) of various performance measures for four methods with hard and soft thresholding rules in Model 2 with gamma-related distributions over 100 replications}\label{table:gamma}
\begin{tabular*}{\textwidth}{@{}l*{8}{@{\extracolsep{\fill}}c}@{}}
\toprule\toprule
& \multicolumn{4}{c}{Hard} & \multicolumn{4}{c@{}}{Soft}\\
\cmidrule{2-5}\cmidrule{6-9}
$p$ & $\what\bOmega$ & $\what\bOmega_o$ & $\what\bOmega_l$ & $\what\bOmega_c$ & $\what\bOmega$ & $\what\bOmega_o$ & $\what\bOmega_l$ & $\what\bOmega_c$\\ \midrule\addlinespace
\multicolumn{9}{@{}c@{}}{Matrix $L_1$-norm loss}\\
~50 & ~4.15 (0.07) & ~4.09 (0.06) & ~92.60 (1.85)~ & ~6.91 (0.00) & ~4.34 (0.06) & ~4.11 (0.06) & ~72.77 (1.45) & ~6.91 (0.00)\\
100 & ~5.45 (0.04) & ~5.44 (0.04) & 159.43 (4.91)~ & ~8.07 (0.00) & ~5.68 (0.05) & ~5.58 (0.05) & 124.90 (3.18) & ~8.07 (0.00)\\
200 & ~8.09 (0.05) & ~7.99 (0.05) & 256.12 (11.01) & 10.93 (0.00) & ~7.98 (0.07) & ~7.95 (0.07) & 200.10 (5.37) & 10.93 (0.00)\\\addlinespace
\multicolumn{9}{@{}c@{}}{Spectral norm loss}\\
~50 & ~2.50 (0.05) & ~2.38 (0.05) & ~68.27 (1.51)~ & ~4.92 (0.00) & ~2.53 (0.02) & ~2.43 (0.02) & ~51.83 (1.17) & ~4.92 (0.00)\\
100 & ~3.25 (0.05) & ~3.19 (0.05) & 111.79 (3.66)~ & ~5.46 (0.00) & ~3.07 (0.02) & ~3.03 (0.02) & ~83.24 (2.42) & ~5.46 (0.00)\\
200 & ~3.86 (0.03) & ~3.87 (0.02) & 170.37 (7.79)~ & ~6.43 (0.00) & ~3.94 (0.02) & ~3.91 (0.02) & 122.81 (4.05) & ~6.43 (0.00)\\\addlinespace
\multicolumn{9}{@{}c@{}}{Frobenius norm loss}\\
~50 & ~6.17 (0.06) & ~5.96 (0.06) & ~70.52 (1.46)~ & 25.98 (0.00) & ~8.82 (0.03) & ~8.45 (0.04) & ~54.44 (1.12) & 25.99 (0.00)\\
100 & ~9.40 (0.06) & ~9.32 (0.06) & 117.87 (3.51)~ & 38.38 (0.00) & 13.92 (0.03) & 13.67 (0.04) & ~90.22 (2.30) & 38.38 (0.00)\\
200 & 13.55 (0.08) & 13.54 (0.09) & 185.38 (7.65)~ & 55.78 (0.00) & 21.64 (0.04) & 21.45 (0.04) & 140.56 (3.83) & 55.78 (0.00)\\\addlinespace
\multicolumn{9}{@{}c@{}}{True positive rate}\\
~50 & ~0.65 (0.01) & ~0.68 (0.01) & ~~0.99 (0.00)~ & ~0.76 (0.02) & ~0.94 (0.01) & ~0.95 (0.00) & ~~0.99 (0.00) & ~0.93 (0.00)\\
100 & ~0.60 (0.00) & ~0.61 (0.00) & ~~0.97 (0.01)~ & ~0.39 (0.02) & ~0.91 (0.00) & ~0.92 (0.00) & ~~0.93 (0.01) & ~0.89 (0.01)\\
200 & ~0.60 (0.00) & ~0.61 (0.00) & ~~0.94 (0.01)~ & ~0.28 (0.02) & ~0.84 (0.00) & ~0.84 (0.00) & ~~0.93 (0.00) & ~0.88 (0.01)\\\addlinespace
\multicolumn{9}{@{}c@{}}{False positive rate}\\
~50 & ~0.00 (0.00) & ~0.00 (0.00) & ~~0.98 (0.01)~ & ~0.48 (0.03) & ~0.12 (0.00) & ~0.11 (0.00) & ~~0.95 (0.00) & ~0.72 (0.01)\\
100 & ~0.00 (0.00) & ~0.00 (0.00) & ~~0.94 (0.02)~ & ~0.10 (0.01) & ~0.07 (0.00) & ~0.07 (0.00) & ~~0.92 (0.01) & ~0.65 (0.01)\\
200 & ~0.00 (0.00) & ~0.00 (0.00) & ~~0.86 (0.03)~ & ~0.06 (0.01) & ~0.04 (0.00) & ~0.04 (0.00) & ~~0.86 (0.02) & ~0.61 (0.01)\\
\bottomrule
\end{tabular*}
\end{sidewaystable}

To further compare the support recovery performance without selecting a threshold level, we plot the receiver operating characteristic (ROC) curves for all methods in Figure \ref{fig:roc}. Note that hard and soft thresholding rules lead to the same ROC curve for each method. We observe that the ROC curves for $\what\bOmega$ and $\what\bOmega_o$ are almost indistinguishable and uniformly dominate those for $\what\bOmega_l$ and $\what\bOmega_c$, demonstrating the superiority of the COAT method. Of the two naive thresholding estimators, $\what\bOmega_l$ tends to outperform $\what\bOmega_c$ when the threshold level is high, since the former is less influenced by the high spurious correlations as reflected in Figure \ref{fig:boxplot}.

\begin{figure}
\includegraphics[width=\textwidth]{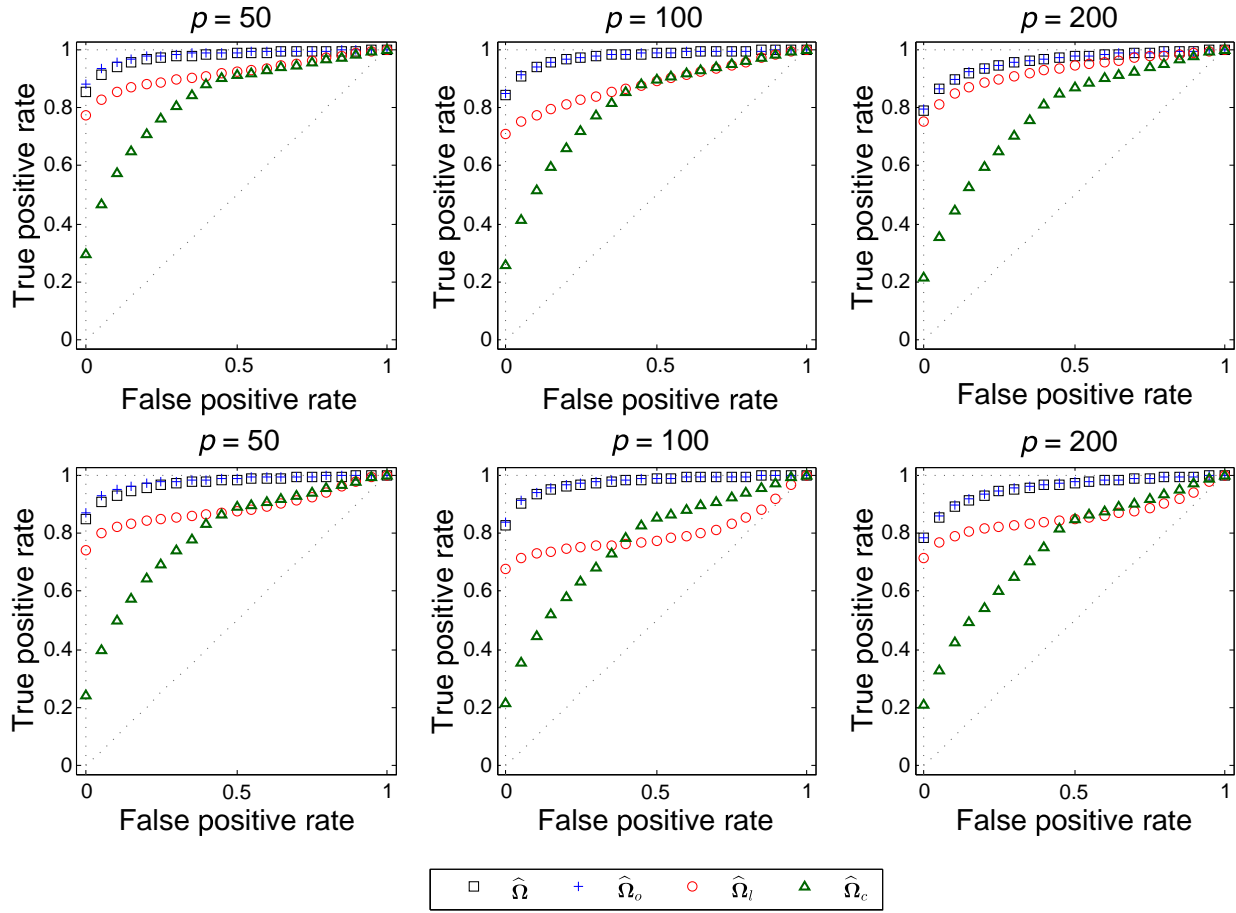}
\caption{ROC curves for four methods in Model 2 with normal-related distribution (top panel) and gamma-related distribution (bottom panel).}\label{fig:roc}
\end{figure}

\section{Gut Microbiome Data Analysis}
The gut microbiome plays a critical role in energy extraction from the diet and interacts with the immune system to exert a profound influence on human health and disease. Despite an emerging interest in characterizing the ecology of human-associated microbial communities, the complex interactions among microbial taxa remain poorly understood \citep{Coyt:Schl:Fost:ecol:2015}. We now illustrate the proposed method by applying it to a human gut microbiome dataset described by \citet{Wu:Chen:Hoff:Bitt:Chen:Keil:link:2011}, which was collected from a cross-sectional study of 98 healthy individuals at the University of Pennsylvania. DNA from stool samples of these subjects were analyzed by 454/Roche pyrosequencing of 16S rRNA gene segments, resulting in an average of 9265 reads per sample, with a standard deviation of 3864. Taxonomic assignment yielded 3068 operational taxonomic units, which were further combined into 87 genera that appeared in at least one sample. Demographic information, including body mass index (BMI), was also collected from the subjects. We are interested in identifying and comparing the correlation structures among bacterial genera between lean and obese subjects. We therefore divided the dataset into a lean group ($\mathrm{BMI}<25$, $n=63$) and an obese group ($\mathrm{BMI}\ge25$, $n=35$), and focused on the $p=40$ bacterial genera that appeared in at least four samples in each group. The count data were transformed into compositions after zero counts were replaced by 0.5.

We applied the COAT method with the soft thresholding rule to each group, and used tenfold cross-validation to select the tuning parameter. The resulting estimate was represented by a correlation network among the bacterial genera with each edge representing a nonzero correlation. To assess the stability of support recovery, we further generated 100 bootstrap samples for each group and repeated the thresholding procedure on each sample. The stability of the correlation network was measured by the average proportion of edges reproduced by each bootstrap replicate. Finally, we retained only the edges in the correlation network that were reproduced in at least 80 bootstrap replicates. The numbers of positive and negative correlations and the stability of correlation networks are reported in Table \ref{table:combo}; the results for the two naive thresholding estimators $\what\bOmega_l$ and $\what\bOmega_c$ are also included for comparison. We see that the COAT method achieves the highest stability among the three methods and has the most edges passing the stability test. The correlation network identified by $\what\bOmega_l$ has substantially fewer negative correlations than the other two methods, which is likely due to the severe upward bias observed in Figure \ref{fig:boxplot}. The correlation network identified by $\what\bOmega_c$ is the least stable.

\begin{table}
\def~{\phantom{0}}
\caption{Numbers of positive and negative correlations and stability of correlation networks for three methods applied to the gut microbiome data}\label{table:combo}
\begin{tabular*}{\textwidth}{@{}l*{6}{@{\extracolsep{\fill}}c}@{}}
\toprule\toprule
& \multicolumn{3}{c}{Lean} & \multicolumn{3}{c@{}}{Obese}\\
\cmidrule{2-4}\cmidrule{5-7}
& $\what\bOmega$ & $\what\bOmega_l$ & $\what\bOmega_c$ & $\what\bOmega$ & $\what\bOmega_l$ & $\what\bOmega_c$\\
\midrule\addlinespace
Positive correlations &  111 &  108 &  119 &   41 &   34 &   31\\
Negative correlations &  134 &   55 &   95 &   55 &   11 &   43\\
Network stability     & 0.83 & 0.68 & 0.67 & 0.87 & 0.62 & 0.54\\
\bottomrule
\end{tabular*}
\end{table}

The correlation networks identified by the COAT method for the two groups are displayed in Figure \ref{fig:net}. Clearly, the networks for the lean and obese groups show markedly different architecture, indicating that the obese microbiome is less modular with less complex interactions between the modules. This phenomenon has been demonstrated by previous studies and is possibly due to adaptation of the microbiome to low-diversity environments \citep{Gree:Turn:Bore:meta:2012}. Table \ref{table:combo} and Figure \ref{fig:net} also suggest that the gut microbial network tends to contain more competitive (negative) interactions than cooperative (positive) ones, which seems consistent with the recent finding that the ecological stability of the gut microbiome can be attributed to the benefits from limiting positive feedbacks and dampening cooperative networks \citep{Coyt:Schl:Fost:ecol:2015}.

\begin{figure}\centering
  \begin{subfigure}{.9\textwidth}
    \includegraphics[width=\textwidth]{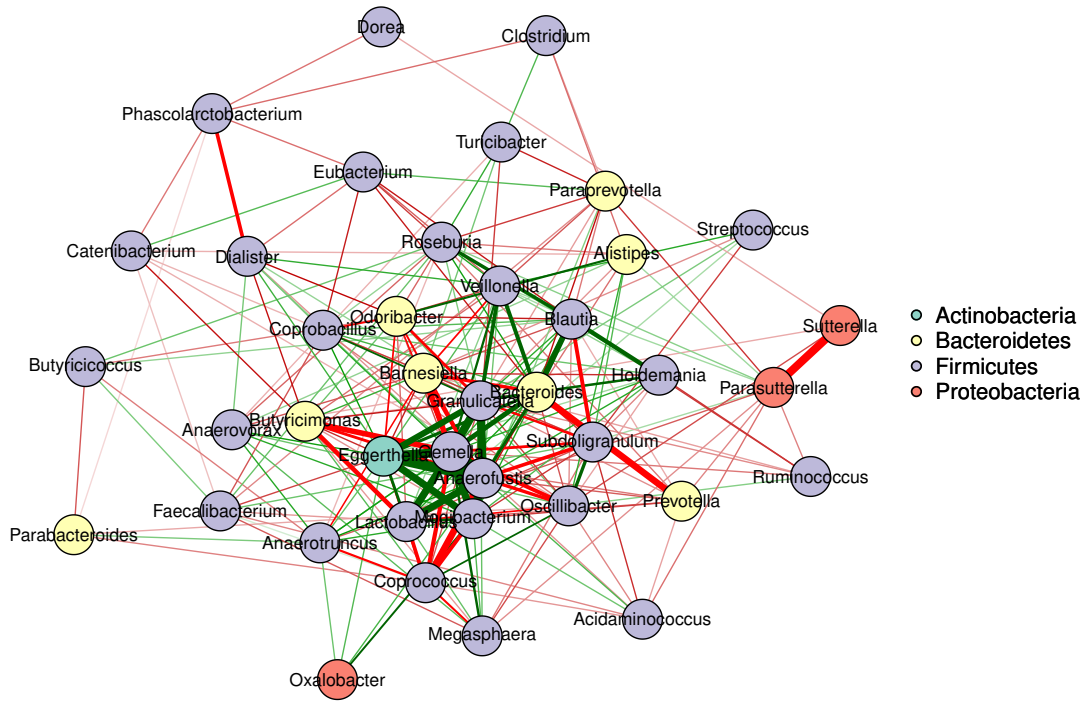}
    \caption{Lean}
  \end{subfigure}\\
  \begin{subfigure}{.9\textwidth}
    \includegraphics[width=\textwidth]{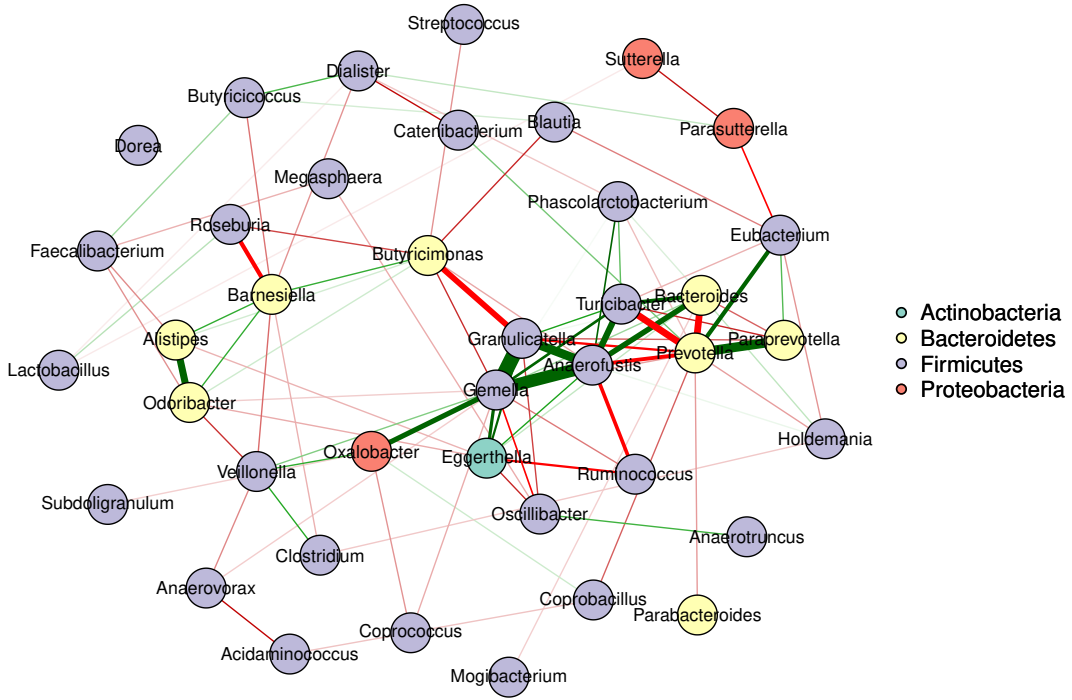}
    \caption{Obese}
  \end{subfigure}
  \caption{Correlation networks identified by the COAT method for the lean and obese groups in the gut microbiome data. Positive and negative correlations are displayed in green and red, respectively. The thickness of edges indicates the magnitude of correlations.}\label{fig:net}
\end{figure}

A closer inspection of the correlation networks identifies \emph{Bacteroides} and \emph{Prevotella} as two key genera of the gut microbiome. The abundances of these two genera are well known to distinguish two gut microbial enterotypes, which are strongly associated with long-term dietary patterns \citep{Arum:Raes:Ehrl:Bork:ente:2011,Wu:Chen:Hoff:Bitt:Chen:Keil:link:2011}. The negative correlations between \emph{Bacteroides} and \emph{Prevotella} ($-0.404$ in the lean group and $-0.296$ in the obese group) are well explained by the diet-dependent enterotypes and the within-body separation of the two genera \citep{jord:Laur:Mori:Pria:dive:2015}. Moreover, recent studies have suggested several keystone species belonging to the genus \emph{Bacteroides}, through which the structure of gut microbial communities may be influenced by small perturbations \citep{Fish:Meht:iden:2014}. Also, the Firmicutes-enriched microbiome has been found to hold greater metabolic potential than the Bacteroidetes-enriched microbiome for more efficient energy harvest from the diet \citep{Turn:Ley:Mard:Gord:an:2006}. Figure \ref{fig:net} seems to support these findings, in view of the central position of \emph{Bacteroides} in the networks and its strong correlations with a few genera belonging to the Firmicutes. Such patterns, however, are less clearly seen in the correlation networks identified by the other two methods.

\section{Discussion}
Understanding the dependence structure among microbial taxa within a community, including co-occurrence and co-exclusion relationships between microbial taxa, is an important problem in microbiome research. Such structures provide biological insights into the community dynamics and factors that change the community structures. To overcome the difficulties arising from the unit-sum constraint of the observed compositional data, we have developed a COAT method to estimate the sparse covariance matrix of the latent log-basis components. Our method is based on a decomposition of the variation matrix into a rank-2 component and a sparse component. The resulting procedure is equivalent to thresholding the sample centered log-ratio covariance matrix, and thus is optimization-free and scalable for high-dimensional data.

Our simulation results demonstrate that the COAT method performs almost as well as the oracle thresholding estimator that knew the latent basis components, and outperforms some naive thresholding estimators by a large margin. These improvements are more pronounced when the basis components have a skewed distribution, as is often observed in microbiome studies. In the application to gut microbiome data, the COAT method leads to more stable and biologically more interpretable results for comparing the dependence structures of lean and obese microbiomes.

We have provided conditions for the approximate and exact identifiability of the covariance parameters, and have established rates of convergence and support recovery guarantees for the COAT estimator. The rate of convergence includes an extra term of $O_p(s_0(p)(s_0(p)/p)^{1-q})$ in addition to the usual minimax optimal rate of convergence for sparse covariance estimation. The extra term represents an approximation error due to using $\bGamma_0$ as a proxy for $\bOmega_0$, which vanishes under mild assumptions as the dimensionality increases.

The proposed methodology may be extended in several ways. First, it would be possible to develop a joint optimization procedure based on the decomposition \eqref{eq:decomp}. For example, one may consider the regularized estimator
\[
\what\bOmega_{\mathrm{reg}}=\argmin_{\bOmega}\{\|\what\bT-\bomega\bone^T-\bone\bomega^T+2\bOmega\|_F^2+P_{\lambda}(\bOmega)\},
\]
where $\bomega=\diag(\bOmega)$ and $P_{\lambda}(\cdot)$ is a sparsity-inducing penalty function. The COAT estimator can be viewed as a one-step approximation to $\what\bOmega_{\mathrm{reg}}$ with appropriately chosen penalty function and initial value $\what\bOmega=\bzero$. Solving the full optimization problem is computationally more expensive but is expected to improve on the performance of the COAT estimator. Another worthwhile extension would be to deal with zero counts directly. One may, in principle, combine the ideas presented here with models that account for sampling and structural zeros. The issues of identifiability and computational feasibility are the major concerns with such extensions.

\appendix
\numberwithin{equation}{section}

\section*{Appendix: Proofs}
\addtocounter{section}{1}
\setcounter{equation}{0}

\subsection{Proof of Proposition \ref{prop:ident}}
Using the fact that the centered log-ratio covariance matrix $\bGamma_0$ is symmetric and has all zero row sums \citep[Property 4.6]{Aitc:stat:2003}, we have
\[
\tr\{(\bgamma_0\bone^T+\bone\bgamma_0^T)^T\bGamma_0\}=\tr(\bgamma_0^T\bGamma_0\bone)+\tr(\bgamma_0\bone^T\bGamma_0)=0,
\]
that is, the components $\bgamma_0\bone^T+\bone\bgamma_0^T$ and $\bGamma_0$ are orthogonal to each other.

To show the desired inequality, by the identity (4.35) of \citet{Aitc:stat:2003}, we have
\[
\omega_{ij}^0-\gamma_{ij}^0=\omega_{ij}^0-(\omega_{ij}^0-\omega_{i\cdot}^0-\omega_{j\cdot}^0+\omega_{\cdot\cdot}^0) =\omega_{i\cdot}^0+\omega_{j\cdot}^0-\omega_{\cdot\cdot}^0.
\]
Therefore,
\[
\|\bOmega_0-\bGamma_0\|_{\max}\le\max_{i,j}(|\omega_{i\cdot}^0|+|\omega_{j\cdot}^0|+|\omega_{\cdot\cdot}^0|)\le3p^{-1}\|\bOmega_0\|_1.
\]

\subsection{Proof of Proposition \ref{prop:ident_exact}}
We first claim that if $\balpha=(\alpha_1,\dots,\alpha_p)^T\ne\bzero$, then the matrix $\bA\equiv\balpha\bone^T+\bone\balpha^T$ has at least $p-1$ nonzero upper-triangular entries. To prove this, without loss of generality, assume $\alpha_1\ne0$ and that the last $q$ entries of the first row of $\bA$ are zero, where $0\le q\le p-1$; that is, $\alpha_1+\alpha_j\ne0$ for $1\le j\le p-q$, and $\alpha_1+\alpha_{p-q+1}=\cdots=\alpha_1+\alpha_p=0$. The latter implies $\alpha_{p-q+1}=\cdots=\alpha_p=-\alpha_1\ne0$, which gives rise to $\binom{q}{2}=q(q-1)/2$ nonzero entries at positions $(i,j)$ with $p-q+1\le i<j\le p$. Putting these pieces together, we obtain that the number of nonzero upper-triangular entries in $\bA$ is at least
\[
f(q)\equiv p-q-1+\frac{q(q-1)}{2}\ge f(1)=f(2)=p-2.
\]
To show that the lower bound $p-2$ is not attainable, note that if there are only $p-2$ nonzero upper-triangular entries, then $q=1$ or 2, and we have $\alpha_2+\alpha_p=\cdots=\alpha_{p-2}+\alpha_p=0$, which implies $\alpha_2=\cdots=\alpha_{p-2}=-\alpha_p=\alpha_1\ne0$. Since $p\ge5$, this gives rise to at least one nonzero entry at positions $(i,j)$ with $2\le i<j\le p-2$, which is a contradiction.

Now suppose $s_e(p)<(p-1)/2$ and that $\bOmega_1$ and $\bOmega_2$ in $\mathcal{B}_0(s_e(p))$ lead to $\bT_1=\bT_2$, that is,
\[
(\bomega_1-\bomega_2)\bone^T+\bone(\bomega_1-\bomega_2)^T=2(\bOmega_1-\bOmega_2).
\]
Note that the right-hand side has fewer than $p-1$ nonzero upper-triangular entries. Then it follows from the above claim that $\bOmega_1=\bOmega_2$.

We prove the other direction by showing that, if $s_e(p)\ge(p-1)/2$, then there exist $\bOmega_1$ and $\bOmega_2$ in $\mathcal{B}_0(s_e(p))$ with $\bOmega_1\ne\bOmega_2$ that lead to $\bT_1=\bT_2$. Indeed, let
\[
\bOmega_1=\begin{pmatrix}
          1+c & c\bone_{p_1}^T &  \bzero_{p_2}^T\\
 c\bone_{p_1} &            \bI &          \bzero\\
 \bzero_{p_2} &         \bzero &             \bI
\end{pmatrix},\quad
\bOmega_2=\begin{pmatrix}
          1-c & \bzero_{p_1}^T & -c\bone_{p_2}^T\\
 \bzero_{p_1} &            \bI &          \bzero\\
-c\bone_{p_2} &         \bzero &             \bI
\end{pmatrix},
\]
where $p_1=\lfloor(p-1)/2\rfloor$, $p_2=p-1-p_1$, and $0<|c|<1$. Then it is easy to verify that
\[
\bT_1=\bT_2=\begin{pmatrix}
               0 &              (2-c)\bone_{p_1}^T &              (2+c)\bone_{p_2}^T\\
(2-c)\bone_{p_1} & 2(\bone_{p_1}\bone_{p_1}^T-\bI) &       2\bone_{p_1}\bone_{p_2}^T\\
(2+c)\bone_{p_2} &       2\bone_{p_2}\bone_{p_1}^T & 2(\bone_{p_2}\bone_{p_2}^T-\bI)\\
\end{pmatrix}.
\]
This completes the proof.

\subsection{Concentration Inequalities}
To prepare for the proofs of Theorems 1 and 2, we first establish some useful concentration inequalities. For notational simplicity, the constants $C_1,C_2,\dots$ below may vary from line to line.

\begin{lemma}\label{lem:conc1}
Under Condition \ref{cond:tail}, there exist constants $C_1,C_2>0$ such that
\begin{equation}\label{eq:conc1}
P\left(\max_j\left|\frac{1}{n}\sum_{k=1}^nY_{kj}\right|\ge t\right)\le C_1pe^{-C_2nt^2}\\
\end{equation}
and
\begin{equation}\label{eq:conc2}
P\left(\max_{i,j}\left|\frac{1}{n}\sum_{k=1}^nY_{ki}Y_{kj}-EY_iY_j\right|\ge t\right)\le C_1p^2e^{-C_2nt^2}
\end{equation}
for sufficiently small $t>0$. Moreover, if $\log p=o(n^{1/5})$, then there exists a constant $C_3>0$ such that
\begin{equation}\label{eq:conc4}
P\left(\max_{i,j,\ell,m}\left|\frac{1}{n}\sum_{k=1}^nY_{ki}Y_{kj}Y_{k\ell}Y_{km}-EY_iY_jY_{\ell}Y_m\right|\ge\ve\right)=O(p^{-C_3})
\end{equation}
for every constant $\ve>0$.
\end{lemma}

\begin{proof}
Inequalities \eqref{eq:conc1} and \eqref{eq:conc2} follow, for example, from Exercise 2.27 of \citet{Bouc:Lugo:Mass:conc:2013}; see also \citet{Bick:Levi:cova:2008}.

To prove \eqref{eq:conc4}, let $Z_{kijlm}=Y_{ki}Y_{kj}Y_{k\ell}Y_{km}$ and $Z_{ijlm}=Y_iY_jY_{\ell}Y_m$. Note first that, by Condition \ref{cond:tail} and the sub-Gaussian tail bound, for any $K>0$ and $i,j,\ell,m$,
\[
P(|Z_{ijlm}|>K)\le4P(|Y_j|>K^{1/4})\le8e^{-\alpha\sqrt{K}/8}.
\]
Hence,
\begin{align*}
E|Z_{ijlm}|I(|Z_{ijlm}|>K)&=\int_0^{\infty}P(|Z_{ijlm}|I(|Z_{ijlm}|>K)>z)\,dz\\
&=KP(|Z_{ijlm}|>K)+\int_K^{\infty}P(|Z_{ijlm}|>z)\,dz\\
&\le8Ke^{-\alpha\sqrt{K}/8}+\int_K^{\infty}8e^{-\alpha\sqrt{z}/8}\,dz\\
&=\frac{8}{\alpha^2}(\alpha^2K+16\alpha\sqrt{K}+128)e^{-\alpha\sqrt{K}/8},
\end{align*}
which is less than $\ve/4$ if we choose $K$ sufficiently large. Then we have
\begin{align*}
&P\left(\max_{i,j,\ell,m}\left|\frac{1}{n}\sum_{k=1}^nZ_{kijlm}-EZ_{ijlm}\right|\ge\ve\right)\\
&\quad\le P\left(\max_{i,j,\ell,m}\left|\frac{1}{n}\sum_{k=1}^nZ_{kijlm}I(|Z_{kijlm}|\le K)-EZ_{ijlm}I(|Z_{ijlm}|\le K)\right|\ge\frac{\ve}{2}\right)\\
&\quad\relphantom{\le}{}+P\left(\max_{i,j,\ell,m}\left|\frac{1}{n}\sum_{k=1}^nZ_{kijlm}I(|Z_{kijlm}|>K)\right|\ge\frac{\ve}{4}\right)\\
&\quad\equiv T_1+T_2.
\end{align*}
By Hoeffding's inequality and the union bound,
\[
T_1\le2p^4\exp\left(-\frac{n\ve^2}{8K^2}\right).
\]
Also, by Condition 1 and the sub-Gaussian tail bound,
\[
T_2\le P\left(\max_{k,i,j,\ell,m}|Z_{kijlm}|>K\right)\le P\left(\max_{k,j}|Y_{kj}|>K^{1/4}\right)\le2npe^{-\alpha\sqrt{K}/8}.
\]
Combining both terms, choosing $K=C^2(\log p+\log n)^2$ with $C>8/\alpha$, and noting $\log p=o(n^{1/5})$, we arrive at
\begin{align*}
&P\left(\max_{i,j,\ell,m}\left|\frac{1}{n}\sum_{k=1}^nZ_{kijlm}-EZ_{ijlm}\right|\ge\ve\right)\\
&\quad\le2p^4\exp\left(-\frac{n\ve^2}{8C^4(\log p+\log n)^4}\right)+2(np)^{1-C\alpha/8}\\
&\quad=O(p^{-C_3})
\end{align*}
for some $C_3>0$. This proves \eqref{eq:conc4} and completes the proof.
\end{proof}

\begin{lemma}\label{lem:conc2}
Under Conditions \ref{cond:tail}--\ref{cond:moment4}, there exist constants $C_1,C_2,C_3>0$ such that
\begin{equation}\label{eq:conc_theta}
P\left(\max_{i,j}|\hat\theta_{ij}-\theta_{ij}|\ge\ve\right)=O(p^{-C_3})
\end{equation}
and
\begin{equation}\label{eq:conc_gamma}
P\left(\max_{i,j}|\hat\gamma_{ij}-\omega_{ij}^0|/\sqrt{\hat\theta_{ij}}\ge C_1\sqrt{\frac{\log p}{n}}+C_2\frac{s_0(p)}{p}\right)=O(p^{-C_3})
\end{equation}
for every constant $\ve>0$.
\end{lemma}

\begin{proof}
We first prove \eqref{eq:conc_theta}. Define
\[
\tilde\theta_{ij}=\frac{1}{n}\sum_{k=1}^n(\gamma_{ki}\gamma_{kj}-\tilde\gamma_{ij})^2,
\]
where $\tilde\gamma_{ij}=n^{-1}\sum_{k=1}^n\gamma_{ki}\gamma_{kj}$. We then write
\begin{align}
\hat\theta_{ij}-\tilde\theta_{ij}&=\frac{1}{n}\sum_{k=1}^n\{(\gamma_{ki}\gamma_{kj}-\tilde\gamma_{ij})-\gamma_{ki}\bar\gamma_j-\gamma_{kj}\bar\gamma_i +2\bar\gamma_i\bar\gamma_j\}^2-\frac{1}{n}\sum_{k=1}^n(\gamma_{ki}\gamma_{kj}-\tilde\gamma_{ij})^2\notag\\
&=\frac{2}{n}\sum_{k=1}^n(\gamma_{ki}\gamma_{kj}-\tilde\gamma_{ij})(-\gamma_{ki}\bar\gamma_j-\gamma_{kj}\bar\gamma_i +2\bar\gamma_i\bar\gamma_j) +\frac{1}{n}\sum_{k=1}^n(-\gamma_{ki}\bar\gamma_j-\gamma_{kj}\bar\gamma_i +2\bar\gamma_i\bar\gamma_j)^2.\label{eq:theta_diff}
\end{align}
Note that, by definition, $\gamma_{kj}=Y_{kj}-\bar{Y}_k$, where $\bar{Y}_k=p^{-1}\sum_{j=1}^pY_{kj}$. Define $\gamma_j=Y_j-\bar{Y}$, where $\bar{Y}=p^{-1}\sum_{j=1}^pY_j$. Since $Y_j$ are uniformly sub-Gaussian by Condition \ref{cond:tail}, $\gamma_j$ are also uniformly sub-Gaussian. Using a truncation argument similar to that for proving \eqref{eq:conc4}, we can show that
\[
P\left(\max_{i,j}\left|\frac{1}{n}\sum_{k=1}^n\gamma_{ki}^2\gamma_{kj}-E\gamma_i^2\gamma_j\right|\ge C_1\right)=O(p^{-C_3})
\]
for some $C_1,C_3>0$. The sub-Gaussian tails imply also that $E\gamma_i^2|\gamma_j|\le\frac{1}{2}(E\gamma_i^4+E\gamma_j^2)=O(1)$. Combining these two pieces yields
\[
P\left(\max_{i,j}\left|\frac{1}{n}\sum_{k=1}^n\gamma_{ki}^2\gamma_{kj}\right|\ge C_1\right)=O(p^{-C_3}).
\]
It follows from Lemma \ref{lem:conc1} that
\[
P\left(\max_j|\bar\gamma_j|\ge C_1\sqrt{\frac{\log p}{n}}\right)=O(p^{-C_3}).
\]
The above two inequalities together imply
\begin{equation}\label{eq:conc_bar}
P\left(\max_{i,j}\left|\frac{1}{n}\sum_{k=1}^n\gamma_{ki}^2\gamma_{kj}\bar\gamma_j\right|\ge C_1\sqrt{\frac{\log p}{n}}\right)=O(p^{-C_3}).
\end{equation}
We can similarly bound the other terms in \eqref{eq:theta_diff} and obtain
\begin{equation}\label{eq:theta_tilde}
P\left(\max_{i,j}|\hat\theta_{ij}-\tilde\theta_{ij}|\ge C_1\sqrt{\frac{\log p}{n}}\right)=O(p^{-C_3}).
\end{equation}

Next, write
\begin{align*}
\tilde\theta_{ij}-\theta_{ij}&=\frac{1}{n}\sum_{k=1}^n(\gamma_{ki}\gamma_{kj}-\tilde\gamma_{ij})^2-\Var(Y_iY_j)\\
&=\frac{1}{n}\sum_{k=1}^n\gamma_{ki}^2\gamma_{kj}^2-EY_i^2Y_j^2-\{\tilde\gamma_{ij}^2-(\omega_{ij}^0)^2\}\\
&\equiv T_1+T_2.
\end{align*}
To bound the term $T_1$, we further write
\begin{align*}
T_1&=\frac{1}{n}\sum_{k=1}^n\{(Y_{ki}-\bar{Y}_k)(Y_{kj}-\bar{Y}_k)\}^2-EY_i^2Y_j^2\\
&=\frac{1}{n}\sum_{k=1}^n\left(Y_{ki}Y_{kj}-Y_{ki}\bar{Y}_k-Y_{kj}\bar{Y}_k+\bar{Y}_k^2\right)^2-EY_i^2Y_j^2\\
&=\frac{1}{n}\sum_{k=1}^nY_{ki}^2Y_{kj}^2-EY_i^2Y_j^2+\frac{2}{n}\sum_{k=1}^nY_{ki}Y_{kj}(-Y_{ki}\bar{Y}_k-Y_{kj}\bar{Y}_k+\bar{Y}_k^2)\\ &\relphantom{=}{}+\frac{1}{n}(-Y_{ki}\bar{Y}_k-Y_{kj}\bar{Y}_k+\bar{Y}_k^2)^2.
\end{align*}
Consider the event $A_1$ on which
\[
\max_{i,j,\ell,m}\left|\frac{1}{n}\sum_{k=1}^nY_{ki}Y_{kj}Y_{k\ell}Y_{km}-EY_iY_jY_{\ell}Y_m\right|\le\ve_1.
\]
Then, on $A_1$, we have
\[
\left|\frac{1}{n}\sum_{k=1}^nY_{ki}^2Y_{kj}^2-EY_i^2Y_j^2\right|\le\ve_1.
\]
To bound the next term in $T_1$, we write
\begin{align*}
\frac{1}{n}\sum_{k=1}^nY_{ki}^2Y_{kj}\bar{Y}_k&=\frac{1}{n}\sum_{k=1}^nY_{ki}^2Y_{kj}\bar{Y}_k-EY_i^2Y_j\bar{Y}+EY_i^2Y_j\bar{Y}\\
&=\frac{1}{p}\sum_{\ell=1}^p\left(\frac{1}{n}\sum_{k=1}^nY_{ki}^2Y_{kj}Y_{k\ell}-EY_i^2Y_jY_\ell\right)+\frac{1}{p}\sum_{\ell=1}^pEY_i^2Y_jY_\ell,
\end{align*}
which, on $A_1$ and by Condition \ref{cond:moment4}, is bounded by $\ve_1+s_1(p)/p$. We can similarly bound the other terms in $T_1$ and obtain, on $A_1$,
\begin{equation}\label{eq:t1}
|T_1|\le16\ve_1+15s_1(p)/p.
\end{equation}

To bound the term $T_2$, note that
\begin{align}
\tilde\gamma_{ij}-\omega_{ij}^0&=\frac{1}{n}\sum_{k=1}^n(Y_{ki}-\bar{Y}_k)(Y_{kj}-\bar{Y}_k)-EY_iY_j\notag\\
&=\frac{1}{n}\sum_{k=1}^nY_{ki}Y_{kj}-EY_iY_j+\frac{1}{n}\sum_{k=1}^n(-Y_{ki}\bar{Y}_k-Y_{kj}\bar{Y}_k+\bar{Y}_k^2).\label{eq:t2diff}
\end{align}
Consider the event $A_2$ on which
\[
\max_{i,j}\left|\frac{1}{n}\sum_{k=1}^nY_{ki}Y_{kj}-EY_iY_j\right|\le\ve_2.
\]
To bound the next term in \eqref{eq:t2diff}, we write
\begin{align*}
\frac{1}{n}\sum_{k=1}^nY_{ki}\bar{Y}_k&=\frac{1}{n}\sum_{k=1}^nY_{ki}\bar{Y}_k-EY_i\bar{Y}+EY_i\bar{Y}\\
&=\frac{1}{p}\sum_{j=1}^p\left(\frac{1}{n}\sum_{k=1}^nY_{ki}Y_{kj}-EY_iY_j\right)+\frac{1}{p}\sum_{j=1}^p\omega_{ij}^0,
\end{align*}
which, on $A_2$ and by Condition \ref{cond:sparse}, is bounded by $\ve_2+M^{1-q}s_0(p)/p$. We can similarly bound the other terms in \eqref{eq:t2diff} and obtain, on $A_2$,
\begin{equation}\label{eq:t2bound}
|\tilde\gamma_{ij}-\omega_{ij}^0|\le4\ve_2+3M^{1-q}s_0(p)/p.
\end{equation}
Note also that, on $A_2$,
\[
|\tilde\gamma_{ij}+\omega_{ij}^0|\le|\tilde\gamma_{ij}-\omega_{ij}^0|+2|\omega_{ij}^0|\le4\ve_2+3M^{1-q}s_0(p)/p+2M.
\]
Hence, on $A_2$, we have
\begin{equation}\label{eq:t2}
|T_2|=|\tilde\gamma_{ij}-\omega_{ij}^0||\tilde\gamma_{ij}+\omega_{ij}^0|\le(4\ve_2+3M^{1-q}s_0(p)/p)(4\ve_2+3M^{1-q}s_0(p)/p+2M).
\end{equation}

Finally, it follows from Lemma \ref{lem:conc1} that the event $A_1\cap A_2$ occurs with probability at least $1-O(p^{-C_3})$ for all constants $\ve_1,\ve_2>0$ and some constant $C_3>0$. Combining \eqref{eq:theta_tilde}, \eqref{eq:t1}, and \eqref{eq:t2} and noting $\log p=o(n)$, $s_0(p)=o(p)$, and $s_1(p)=o(p)$, we arrive at \eqref{eq:conc_theta}.

It remains to prove \eqref{eq:conc_gamma}. We first write
\begin{align*}
\hat\gamma_{ij}-\tilde\gamma_{ij}&=\frac{1}{n}\sum_{k=1}^n(\gamma_{ki}-\bar\gamma_i)(\gamma_{kj}-\bar\gamma_j)-\frac{1}{n}\sum_{k=1}^n\gamma_{ki}\gamma_{kj}\\
&=\frac{1}{n}\sum_{k=1}^n(-\gamma_{ki}\bar\gamma_i-\gamma_{kj}\bar\gamma_j+\bar\gamma_i\bar\gamma_j).
\end{align*}
Using arguments similar to those for proving \eqref{eq:conc_bar}, we can show that
\[
P\left(\max_{i,j}\left|\frac{1}{n}\sum_{k=1}^n\gamma_{ki}\bar\gamma_j\right|\ge C_1\sqrt{\frac{\log p}{n}}\right)=O(p^{-C_3}).
\]
We can similarly bound the other two terms and obtain
\[
P\left(\max_{i,j}|\hat\gamma_{ij}-\tilde\gamma_{ij}|\ge C_1\sqrt{\frac{\log p}{n}}\right)=O(p^{-C_3}).
\]
Taking $\ve_2=C_1\sqrt{(\log p)/n}$ in \eqref{eq:t2bound}, we have
\[
P\left(\max_{i,j}|\tilde\gamma_{ij}-\omega_{ij}^0|\ge C_1\sqrt{\frac{\log p}{n}}+C_2\frac{s_0(p)}{p}\right)=O(p^{-C_3}).
\]
The above two inequalities together imply
\begin{equation}\label{eq:conc_top}
P\left(\max_{i,j}|\hat\gamma_{ij}-\omega_{ij}^0|\ge C_1\sqrt{\frac{\log p}{n}}+C_2\frac{s_0(p)}{p}\right)=O(p^{-C_3}).
\end{equation}
From Condition \ref{cond:lower} and \eqref{eq:conc_theta} with $\ve_2=\tau/2$, it follows that $|\hat\theta_{ij}|\ge\tau/2$ with probability at least $1-O(p^{-C_3})$. This, together with \eqref{eq:conc_top}, implies \eqref{eq:conc_gamma} and completes the proof.
\end{proof}

\subsection{Proof of Theorem \ref{thm:rate}}
By the triangle inequality, we have
\begin{equation}\label{eq:tri}
\|\what\bOmega-\bOmega_0\|_1\le\sum_{j=1}^p|S_{\lambda_{ij}}(\omega_{ij}^0)-\omega_{ij}^0| +\sum_{j=1}^p|S_{\lambda_{ij}}(\hat\gamma_{ij})-S_{\lambda_{ij}}(\omega_{ij}^0)|.
\end{equation}
Using Conditions (i) and (ii) that define a general thresholding function, the first term above is bounded by
\begin{align*}
&\sum_{j=1}^p|\omega_{ij}^0|I(|\omega_{ij}^0|\le\lambda_{ij})+\sum_{j=1}^p\lambda_{ij}I(|\omega_{ij}^0|>\lambda_{ij})\\
&\quad=\sum_{j=1}^p|\omega_{ij}^0|^q|\omega_{ij}^0|^{1-q}I(|\omega_{ij}^0|\le\lambda_{ij})+\sum_{j=1}^p\lambda_{ij}^q\lambda_{ij}^{1-q} I(|\omega_{ij}^0|>\lambda_{ij})\\
&\quad\le\sum_{j=1}^p|\omega_{ij}^0|^q\lambda_{ij}^{1-q}.
\end{align*}
On the other hand, the second term in \eqref{eq:tri} is bounded by
\begin{align*}
&2\sum_{j=1}^p|\hat\gamma_{ij}|I(|\hat\gamma_{ij}|>\lambda_{ij},|\omega_{ij}^0|\le\lambda_{ij}) +2\sum_{j=1}^p|\omega_{ij}^0|I(|\hat\gamma_{ij}|\le\lambda_{ij},|\omega_{ij}^0|>\lambda_{ij})\\
&\quad{}+\sum_{j=1}^p|S_{\lambda_{ij}}(\hat\gamma_{ij})-S_{\lambda_{ij}}(\omega_{ij}^0)|I(|\hat\gamma_{ij}|>\lambda_{ij},|\omega_{ij}^0|>\lambda_{ij})\\
&\quad\equiv T_1+T_2+T_3.
\end{align*}
To bound the term $T_1$, we write
\begin{align*}
\frac{T_1}{2}&\le\sum_{j=1}^p|\hat\gamma_{ij}-\omega_{ij}^0|I(|\hat\gamma_{ij}|>\lambda_{ij},|\omega_{ij}^0|\le\lambda_{ij}/2)\\ &\relphantom{\le}{}+\sum_{j=1}^p|\hat\gamma_{ij}-\omega_{ij}^0|I(|\hat\gamma_{ij}|>\lambda_{ij},\lambda_{ij}/2<|\omega_{ij}^0|\le\lambda_{ij}) +\sum_{j=1}^p|\omega_{ij}^0|I(|\hat\gamma_{ij}|>\lambda_{ij},|\omega_{ij}^0|\le\lambda_{ij})\\
&\equiv T_4+T_5+T_6.
\end{align*}
Consider the event $B_1$ on which $|\hat\gamma_{ij}-\omega_{ij}^0|\le\lambda_{ij}/2$ for all $i,j$. On $B_1$, we have
\begin{gather*}
T_4\le\sum_{j=1}^p|\hat\gamma_{ij}-\omega_{ij}^0|I(|\hat\gamma_{ij}-\omega_{ij}^0|>\lambda_{ij}/2)=0,\\
T_5\le\sum_{j=1}^p\left(\frac{\lambda_{ij}}{2}\right)^q\left(\frac{\lambda_{ij}}{2}\right)^{1-q} I(|\hat\gamma_{ij}|>\lambda_{ij},\lambda_{ij}/2<|\omega_{ij}^0|\le\lambda_{ij})\le\frac{1}{2^{1-q}}\sum_{j=1}^p|\omega_{ij}^0|^q\lambda_{ij}^{1-q},
\end{gather*}
and
\[
T_6\le\sum_{j=1}^p|\omega_{ij}^0|^q\lambda_{ij}^{1-q}.
\]
Combining these pieces yields
\[
T_1\le2\left(1+\frac{1}{2^{1-q}}\right)\sum_{j=1}^p|\omega_{ij}^0|^q\lambda_{ij}^{1-q}\le4\sum_{j=1}^p|\omega_{ij}^0|^q\lambda_{ij}^{1-q}.
\]
We can similarly bound the terms $T_2$ and $T_3$ on $B_1$:
\begin{align*}
T_2&\le2\sum_{j=1}^p\left(|\hat\gamma_{ij}-\omega_{ij}^0|+|\hat\gamma_{ij}|\right)I(|\hat\gamma_{ij}|\le\lambda_{ij},|\omega_{ij}^0|>\lambda_{ij})\\
&\le2\sum_{j=1}^p\left(\frac{\lambda_{ij}}{2}+\lambda_{ij}\right)I(|\hat\gamma_{ij}|\le\lambda_{ij},|\omega_{ij}^0|>\lambda_{ij}) \le3\sum_{j=1}^p|\omega_{ij}^0|^q\lambda_{ij}^{1-q},\\
T_3&\le\sum_{j=1}^p\left(|\hat\gamma_{ij}-\omega_{ij}^0|+|S_{\lambda_{ij}}(\hat\gamma_{ij})-\hat\gamma_{ij}| +|S_{\lambda_{ij}}(\omega_{ij}^0)-\omega_{ij}^0|\right)I(|\hat\gamma_{ij}|>\lambda_{ij},|\omega_{ij}^0|>\lambda_{ij})\\
&\le\sum_{j=1}^p\left(\frac{\lambda_{ij}}{2}+\lambda_{ij}+\lambda_{ij}\right)I(|\hat\gamma_{ij}|>\lambda_{ij},|\omega_{ij}^0|>\lambda_{ij}) \le\frac{5}{2}\sum_{j=1}^p|\omega_{ij}^0|^q\lambda_{ij}^{1-q}.
\end{align*}
Collecting all terms, we obtain, on $B_1$,
\begin{equation}\label{eq:sum}
\|\what\bOmega-\bOmega_0\|_1\le\frac{21}{2}\sum_{j=1}^p|\omega_{ij}^0|^q\lambda_{ij}^{1-q}.
\end{equation}

Next, we consider the event $B_2$ on which $|\hat\theta_{ij}-\theta_{ij}|\le\tau$ for all $i,j$. From Condition \ref{cond:lower} we have, on $B_2$,
\begin{equation}\label{eq:theta_hat}
\hat\theta_{ij}\le|\hat\theta_{ij}-\theta_{ij}|+\theta_{ij}\le\tau+\theta_{ij}\le2\theta_{ij}.
\end{equation}
Note that, by Condition \ref{cond:tail},
\begin{equation}\label{eq:theta_upper}
\theta_{ij}\le EY_i^2Y_j^2\le\frac{1}{2}(EY_i^4+EY_j^4)\le\frac{2}{\alpha^2}.
\end{equation}
Taking $\lambda_{ij}=\lambda\sqrt{\hat\theta_{ij}}$ with $\lambda=C_1\sqrt{(\log p)/n}+C_2s_0(p)/p$ in \eqref{eq:sum} and applying \eqref{eq:theta_hat} and \eqref{eq:theta_upper}, we obtain, on $B_1\cap B_2$,
\[
\|\what\bOmega-\bOmega_0\|_1\le\frac{21}{2}\sum_{j=1}^p|\omega_{ij}^0|^q\lambda^{1-q}\left(\frac{2}{\alpha}\right)^{1-q} \le\frac{21}{\alpha}s_0(p)\left(C_1\sqrt{\frac{\log p}{n}}+C_2\frac{s_0(p)}{p}\right)^{1-q}.
\]
We conclude the proof by noting that the event $B_1\cap B_2$ occurs with probability $1-O(p^{-C_3})$ by Lemma \ref{lem:conc2} and that the spectral norm is bounded by the matrix $L_1$-norm.

\subsection{Proof of Theorem \ref{thm:supp}}
It follows from Condition (i) and \eqref{eq:conc_gamma} that
\begin{align*}
&P\left(\hat\omega_{ij}\ne0,\omega_{ij}^0=0\text{ for some }i,j\right)\le P\left(\max_{i,j}|\hat\gamma_{ij}-\omega_{ij}^0|\ge\lambda_{ij}\right)\notag\\ &\quad=P\left(\max_{i,j}|\hat\gamma_{ij}-\omega_{ij}^0|/\sqrt{\hat\theta_{ij}}\ge C_1\sqrt{\frac{\log p}{n}}+C_2\frac{s_0(p)}{p}\right)=O(p^{-C_3}),
\end{align*}
which proves \eqref{eq:sparsist}.

To prove \eqref{eq:supp}, note that, by Condition (ii),
\[
P\left(\sgn(\hat\omega_{ij})\ne\sgn(\omega_{ij}^0),\omega_{ij}^0\ne0\text{ for some }i,j\right)\le P\left(|\hat\gamma_{ij}-\omega_{ij}^0|\ge|\omega_{ij}^0|-\lambda_{ij}\text{ for some }i,j\right).
\]
Also, by taking $\ve=3\tau/4$ in \eqref{eq:conc_theta}, we have, with probability $1-O(p^{-C_3})$,
\[
\left|\sqrt{\hat\theta_{ij}}-\sqrt{\theta_{ij}}\right|=\frac{|\hat\theta_{ij}-\theta_{ij}|}{\sqrt{\hat\theta_{ij}}+\sqrt{\theta_{ij}}} \le\frac{3\tau/4}{\sqrt{\tau/4}+\sqrt{\tau}}=\frac{\sqrt{\tau}}{2},
\]
and hence
\begin{align*}
|\omega_{ij}^0|-\lambda_{ij}&\ge C\lambda\sqrt{\theta_{ij}}-\lambda\left(\sqrt{\hat\theta_{ij}}-\sqrt{\theta_{ij}}+\sqrt{\theta_{ij}}\right)\\
&\ge (C-1)\lambda\sqrt{\tau}-\lambda\frac{\sqrt{\tau}}{2}=\left(C-\frac{3}{2}\right)\lambda\sqrt{\tau}
\end{align*}
for all $i,j$. Now applying \eqref{eq:conc_top} yields
\[
P\left(\sgn(\hat\omega_{ij})\ne\sgn(\omega_{ij}^0),\omega_{ij}^0\ne0\text{ for some }i,j\right)=O(p^{-C_3}),
\]
which, together with \eqref{eq:sparsist}, proves the result.

\setstretch{1.24}
\bibliographystyle{jasa}
\bibliography{coat}

\end{document}